\documentclass[jmp,preprint]{revtex4-1}
\draft

\usepackage{amsfonts, amsmath, amssymb, amsthm, graphicx,mathtools }
\usepackage{natbib}
\newtheorem{theorem}{Theorem}
\newtheorem*{theorem*}{Theorem}

\newtheorem{Lemma}[theorem]{Lemma}

\usepackage{bbold}
\DeclareMathOperator{\Ker}{Ker}

\begin{document}
\title{Asymptotically Decreasing Lieb-Robinson Velocity for a Class of Dissipative Quantum Dynamics} 
\author{Beno\^it Descamps}%

\email{benoit.descamps@univie.ac.at}
\affiliation{Faculty of Physics, University of Vienna, Austria}%

\date{\today}

\begin{abstract}
We study the velocity of the propagation of information for a class of local dissipative quantum dynamics. This finite velocity is expressed by the so-called Lieb-Robinson bound. Besides the properties of the already studied dynamics, we consider an additional relation that expresses the propagation of certain subspaces. The previously derived bounds did not reflect the dissipative character of the dynamics and yielded the same result as for the reversible case. In this article, we show that for this class the velocity of propagation of information is time dependent and decays in time towards a smaller velocity. In some cases the velocity becomes zero. At the end of the article, the exponential clustering theorem of general frustration free local Markovian dynamics is revisited.
\end{abstract}

\maketitle

\section{Introduction}
Locality and causality are two of the most fundamental concepts in physics. Measurements are realised in a small subset of space during a small interval of time. Two measurements far away from each other should not influence each other for a long period of time. In other words, information propagates with a finite velocity. This is implied by the local Lorentz invariance or covariance. In 1972, \citet{LIEB} showed that this finiteness is inherited by locality and not the property of local symmetry. This result was shown using local Hamiltonians and is known as the Lieb-Robinson bound. Since then the Lieb-Robinson bound has turned out to be a very useful tool in many areas of modern condensed matter such as the approximation of ground states \cite{2007PhRvA..75d2306O}, the study of gapped Hamiltonians and area laws \cite{PhysRevLett.93.140402,2006PhRvB..73h5115H,2007JSMTE..08...24H}, topological properties \cite{PhysRevB.72.045141}, and others \cite{2007PhRvL..99p7201B,2009PhRvL.102a7204H,2009CMaPh.286.1073N,2009JMP....50i5207H}, and even linear algebra \cite{almostcomm}. An overview of many applications such as the study of gapped systems can be found in \cite{2010arXiv1004.2086N,2011arXiv1102.0835N,2010arXiv1008.5137H}. Such a light cone property in matter has been studied experimentally \cite{cheneau2012lig}.
Local dissipative systems have been studied in the context quantum computing \cite{citeulike:5219312,Kraus,2011PhRvL.107l0501K} or study of phase transitions \cite{PhysRevA.87.012108}. 

Recently this bound was studied for systems with local dissipation \cite{2011arXiv1103.1122N,PhysRevLett.104.190401}. The derived bound was very similar to the original result for non-dissipative systems. In this article, we revisit and derive a class of local dynamics with a more refined bound. Local dissipation implies a local de-coherence  of states in time. Imposing an additional constraint, we show that certain systems have a decaying velocity of propagation,i.e. the light cone closes with time. 
In the first part of the article, we define the subclass and prove the new bound. We then give some examples. At the end, the method is prove the conjecture in \cite{BOE}.
 Finally, we discuss the exponential clustering theorem for frustration free local Markovian dynamics.
\section{Local Dissipative Dynamics}
\subsection{Propagation of Information}
This section explains the connection between the Lieb-Robinson bound and propagation of information.
Assume that a system is initially prepared in a state $\rho$. This state, then, evolves time under the dynamics $\Gamma_t$. The fact that any system can be taken to be part of a larger one automatically implies that the map $\Gamma_t$ needs to be complete positive and trace-preserving. By the same argument, any change to a state of a system should be described by a complete positive trace-preserving map $\Phi$. Any complete positive map has a special decomposition \cite{citeulike:7291175},
$$\Phi: \rho \to \Phi[\rho]=\sum_i \operatorname{Tr}(V_i V_i^\dagger \rho) \frac{V_i \rho V_i^\dagger}{ \operatorname{Tr}V_i^\dagger V_i \rho)}$$
This decomposition interprets $\Phi$ as a change to $\rho$ due to the different measurement outcomes $\frac{V_i \rho V_i^\dagger}{ \operatorname{Tr}(V_i V_i^\dagger \rho)}$ with probabilities  $\operatorname{Tr}(V_i V_i^\dagger \rho)$. Any measurement can be described in a similar way using a set of matrices $\{V_i\}$ so that $\sum_i V_i^\dagger V_i=\mathbb{1}$.
Let us now say that at $t=0$, we apply some local change,$\Phi_x$, at some point $x$ in space. After some time $t$,  we make a local measurement,  $V_y^\dagger$,  at another point $y$ of space. A change is detected by a difference in the probabilities of the outcome. Maximizing over all initial states gives us a bound on how local changes are being perceived at other points in space.
This yields us the following, 
\begin{equation}
|\operatorname{Tr}(V_y[.]V_y^\dagger\circ \Gamma_t\circ \Phi_x [\rho])-\operatorname{Tr}(V_y[.]V_y^\dagger \circ \Phi_x [\rho])|\leq \|(\Phi_x^*[.]-[.])\circ \Gamma^*_t[V_y^\dagger V_y]|\
\label{eq1}
\end{equation}
 where $\Phi_x^*, \Gamma_t^*$ are the adjoint with respect to the Hilbert-Schmidt scalar product.
The operator $\Phi_x^*[.]-[.]$ can be shown to be the generator of an unital Markovian dynamics, i.e. $\forall t\geq 0, \exp[t(\Phi_x[.]-[.])]$ is trace-preserving and complete positive. The case $\|(\Phi_x^*[.]-[.])\circ \Gamma^*_t[V_y^\dagger V_y]\|\propto \|[B_x,.]\circ \Gamma^*_t[V_y^\dagger V_y]\|$, is bounded by the so-called the Lieb-Robinson bound stated in equation (\ref{LRB}) below. However, for proving the existence of the thermodynamic limit of local Markovian dynamics, arbitrary local Lindblad generators are needed.

Obviously, for $t=0$ we see that the upper bound in equation (\ref{eq1}) is always zero. A question that is often asked is what happens for $t>0$.
The Lieb-Robinson bound shows that the change is progressive. For a fixed time, as the measurement is done farther away, it becomes harder and harder to perceive the perturbation. Measurements that detect the change at some time $d$, can be made at least at a distance $d+v \Delta t$ after some time $\Delta t$ with the same accuracy of the results.  In other words,  the propagation of the information is inside a light cone.
 In this article, we show that due to dissipation, we get event horizons.
\section{Local Markovian Dynamics}
Consider a $D$-dimensional lattice $\mathbb{Z}^D$ with a metric. At each point $x\in \mathbb{Z}^D$ of the lattice, define a $d$-dimensional Hilbert space $\mathcal{H}_x$ and for each finite set $\Lambda \subset  \mathbb{Z}^D$ denote the product space,
$$\mathcal{H}_\Lambda = \bigotimes\limits_{x\in \Lambda}\mathcal{H}_x$$
Denote the algebra of all matrices acting on $\mathcal{H}_\Lambda$, i.e. the algebra of local observables of $\Lambda$, $\mathcal{A}_\Lambda$,
$$\mathcal{A}_\Lambda = \bigotimes\limits_{x\in \Lambda}\mathcal{B}(\mathcal{H}_x)$$ 
If $\Lambda_1 \subset \Lambda_2$, the algebra $\mathcal{A}_{\Lambda_1}$ can be identified with the algebra $\mathcal{A}_{\Lambda_1}\otimes \mathbb{1}_{\Lambda_1 \backslash \Lambda_2}$ and therefore  $\mathcal{A}_{\Lambda_1} \subset \mathcal{A}_{\Lambda_2}$.
Define the support of a local observable $A\in \mathcal{A}_\Lambda$ as the minimal set $X\subset \Lambda$ for which $A=A'\otimes \mathbb{1}_{\Lambda\backslash X}$.

A natural norm to consider in this type of systems is the so-called completely bounded norm. The cb-norm of an operator $L$ is defined by, 
$$\|L \|_{cb} =\underset{n\geq 1}{\operatorname{sup}}\|L \otimes \operatorname{id}_{\mathcal{M}_n} \| $$
As shown originally, the Lieb-Robinson velocity is proportional to the norm the local interaction. This proportionality can be made more precise using reproducing functions. 
We assume there exists a non-increasing function $F:[0,\infty)\to (0,\infty)$ so that $F$ is uniformly integrable over the lattice,

$$\| F \| := \underset{x\in \mathbb{Z}^d}{\operatorname{sup}}\sum\limits_{y\in \mathbb{Z}^d}F(d(x,y))<\infty$$

and there is a constant $C_{\mu}$,

$$C_\mu := \underset{x,y\in \mathbb{Z}^d}{\operatorname{sup}}\sum\limits_{z\in \mathbb{Z}^d}\frac{F(d(x,y))F(d(y,z)}{F(d(x,y)}e^{-d(x,y)-d(y,z)+d(x,y)}$$
so that $C_{\mu}\big|_{\mu=0}<\infty$.

Further information about reproducing functions can be found in \cite{KomatHastings}.
We study local one-parameter dynamics $\Gamma_t: \mathcal{A}_{\Lambda} \to  \mathcal{A}_{\Lambda}$ so that $\forall A\in  \mathcal{A}_{\Lambda}$,
$$\frac{d}{dt}\Gamma_t (A)\big|_{t=0}=L(A)=\sum_X I_X(A) $$
with local operators $I_X:\mathcal{A}_{\Lambda}\to \mathcal{A}_{\Lambda}$, $\forall A\in\mathcal{A}_Y, I_X[A_Y]=0$ , if $X\cap Y =\emptyset$, and  $\operatorname{diam}(X)\leq R <\infty$.
Additionally we impose that $\forall t\geq 0$ and $\forall X$,
$$\exp[tI_X]:\mathcal{A}_X \to \mathcal{A}_X$$
is identity preserving and complete positive.
Therefore $L$ and each $I_X$ are the Heisenberg picture equivalent of the Lindblad generator \cite{lind1976} of a continuous one-parameter semi-group of complete positive trace-preserving maps. In the theorem, we impose $I_X^*[\mathbb{1}]=0$. Notice that the complete positivity implies $I_X+I_X^* \leq 0$.

Additionally assume there exists a $\mu>0$, so that,

$$\|L\|_{\mu}= \underset{x,y\in \mathbb{Z}^d}{\operatorname{sup}}\sum\limits_{Z \ni x,y} \frac{\|I_Z\|_{cb}}{F(d(x,y)} e^{\mu d(x,y)}<\infty ,\mbox{and }~~\lambda_{I_X+I_X^*}/2 \geq -\nu$$

where $\lambda_{T}$ denotes the largest non-zero eigenvalue of operator $T$ and $\nu >0$.  The existence of the thermodynamic limit of such systems was shown in \cite{2011arXiv1103.1122N}.

Consider $\Lambda_1, \Lambda_2 \subset \mathbb{Z}^D$, $\Lambda_1 \cap \Lambda_2 =\emptyset$ and $A\in \mathcal{A}_{\Lambda_1}$ and $B\in \mathcal{A}_{\Lambda_2}$.
In the case that $\Gamma_t$ is automorphic, it has already been shown that there exists $C',v$ and $\xi$, \cite{LIEB,2006JSP...124....1N,PhysRevLett.104.190401} so that,
\begin{equation}
\|[B,\exp[tL]A]\|\leq C' \|A\|\ \|B\|\ \exp\left[\frac{vt-d(A,B)}{\xi}\right]
\label{LRB}
\end{equation}
These constants depend on the diameter of the support of $A$ and $B$ and the strength of the interaction. This bound, of course,  should also always be understood for small times compared to the distance. 
Notice, that in the case that $L$ has a unique fixed point, in principle by taking the right observable $A$, we should get something like,
$$\|[B,\exp[tL]A]\|\leq C' \|A\|\ \|B\| \exp[-t\operatorname{Gap}(L)]$$
This bound makes sense for all times, but does not contain any information about the distance between observables.
We will reconcile and combine both properties in the following result.
We show that for a certain class of Lindblad generators, the usual light-cone picture is only valid for small $t$. After a certain time, the dissipative character takes over and the velocity $v>0$ becomes time dependent and decreases towards some smaller velocity which can be zero.

Define $P_X$ as the projector onto $\Ker(I_X + I_X^*
)$. In order to derive our desired result, we need the following property,

\begin{align}
P_Y I_X (\mathbb{1}- P_Y)=0,~~P_Y P_X (\mathbb{1}-P_Y)=0,~~~\forall X,Y \subset \mathbb{Z}^d
\label{structure}
\end{align}

Notice that since projectors are self-adjoint, the second equality implies the mutual commutation relations between the projectors,

On the other hand, the first inequality does not necessarily implies that the projectors and the interaction terms commute. If this were the case, it can be easily shown that dynamics becomes ultra-local, i.e. the support of local operators does not grow under time-evolution.

 Per construction, the invariant state of the system is the uniform state. Let us take a local perturbation at the origin that can be dissipated by the interaction term that overlaps with origin. Under time evolution, this perturbation is then propagated by the other interactions. However, the equations (\ref{structure}) tell us that as the perturbation spreads, it is dissipated at a faster rate.
\begin{theorem*}
  Let $L$ be a local Lindblad generator. Given that for each local interaction $I_X$, $I_X^*[\mathbb{1}]=0$ and the equations (\ref{structure}) are satisfied. Then $\forall A\in \mathcal{A}_{\Lambda_A}, B\in \mathcal{A}_{\Lambda_B}$ with $\Lambda_A \cap \Lambda_B=\emptyset$,
\begin{equation}
\|[\exp(tL)A,B]\|\leq 2 \frac{\|F\|}{C_{\mu}}\|A\|\ \|B\|\ \operatorname{min}\{\Lambda_A,\Lambda_B\}\left(\exp\left[\int_0^t v(\tau) d\tau\right]-1\right)\exp \left[ -  \mu d(A,B)\right]
\label{eventh}
\end{equation} 
with the time-dependent velocity $v(t)$,
$$v(t)= \left[\alpha +\beta e^{-\nu t}\right] C_{\mu} \|L\|_{\mu}$$
with $\nu =-\max_X \lambda_{I_X+I_X^*}/2$.
$$ \alpha = \max_X \| I_X P_X  \|_{cb}/\|I_X\|_{cb}~~\mbox{ and }\beta = \max_X\|  I_X (\mathbb{1}-P_X)\|_{cb}/\|I_X\|_{cb}$$
\end{theorem*}

\begin{proof}
The proof goes as follows. First, we find a meaningful integral representation for $[\exp(tL)A,B]$. This representation shows that in order for the two observables $A$ and $B$ to communicate with other, they need to be propagated by the local interactions. Second, we see that the local interactions not only play the role of propagator but they are also dissipators. This role is expressed by the relation (\ref{structure}). If, after some propagations, the observable is being dissipated at a rate $\lambda$, under the action of another interaction term the rate becomes $\lambda + \nu$.
These two ideas are then combined to give the result.

We remind the reader of the property of the exponential of a sum.
\begin{equation}
\label{exp1}
e^{t(A+B)}=e^{tA}+\int_0^t ds\  e^{(t-s)A} B e^{s(A+B)}
\end{equation}
and 
\begin{equation}
\label{exp2}
e^{t(A+B)}=e^{tA}+\int_0^t ds\  e^{s(A+B)} B e^{(t-s)A}
\end{equation}
Originally the dynamics \cite{PhysRevB.69.104431} was recursively expanded using the first property (\ref{exp1}), Instead of writing down the full integral representation, a recursive inequality was used.
 In our case we use the second property (\ref{exp2}).

Define a path of  $n$ steps, $\mathcal{P}_n$, between $A$ and $B$ as the set of the support of the local interactions $I_{\Lambda_j}$,
$$\mathcal{P}_n=\{\Lambda_j|\Lambda_{B} \cap \Lambda_1 \not=\emptyset , \Lambda_1 \cap \Lambda_2 \not=\emptyset,...,\Lambda_n \cap \Lambda_{A}  \not=\emptyset \mbox{ , all other intersections are empty}\}$$ 
Denote $\mathcal{P}_n^k $ as the subset of 
$\mathcal{P}_n$ containing the first $k$ steps of the path, where $\mathcal{P}_n^0 =\emptyset $. 
Define the generator containing local interaction which do not intersect with the path,
$$L^c_{\mathcal{P}_n^k}= L-\sum_{Z\cap \mathcal{P}_n^k\not= \emptyset } I_Z $$

In Lemma 1 we show,
\begin{equation*}
\|[B, \exp(tL)A]\| \leq  \sum_{n=1}^\infty \|L\|_\mu^n \left(\sup\limits_{\mathcal{P}_n}\mathcal{J}_{\mathcal{P}_n}\right) \sum\limits_{\mathcal{P}_n}\sum\limits_{x_0\in \Lambda_B}\prod_{\Lambda_j \in \mathcal{P}_n}\sum_{x_j \in \Lambda_j}\sum\limits_{x_{n+1} \in \Lambda_A}F(d(x_{j-1},x_j)) e^{-\mu d(x_{j-1},x_j)}
\end{equation*}
with,
\begin{align}
&\mathcal{J}_{\mathcal{P}_n} = \Big|\Big|\int_0^t \int_0^{t-s_1
}...\int_0^{t-\sum_{j=1}^{n-1}s_j} d\textbf{s}\ [B,.]  \  \prod_{j=1}^{n}\left(\exp\left[s_j L^c_{\mathcal{P}_n^{j-1}}\right]  I_{\Lambda_j}/\|I_{\Lambda_j}\|_{cb}\right) \exp[(t-\sum_{j=1}^{n}s_j) L^c_{\mathcal{P}_n^n}]A \Big|\Big|
\label{Jpn}
\end{align}
with the ordered product $\prod_{j=1}^n C_j=C_1 C_2 ... C_n$ and the supremum is taken over all paths of length $n$.

Further on we show for fixed $n$, 
\begin{align*}
&\mathcal{J}_{\mathcal{P}_n} \leq 2 \|A\|\ \|B\|\frac{\left[\max_X \| I_X P_X\|/\|I_{X}\|_{cb}+ \max_X\|I_X (\mathbb{1}-P_X)\|/\|I_{X}\|_{cb} (1-e^{-\nu t})/(\nu t)\right]^n }{n!} t^n
\end{align*}

From the definition of $C_\mu$ and the triangle inequality, it can be seen that,
\begin{equation*}
\sum_{\mathcal{P}_n}\sum_{x_0\in \Lambda_B}\prod_{\Lambda_j \in \mathcal{P}_n}\sum_{x_j \in \Lambda_j}\sum_{x_{n+1} \in \Lambda_A}F(d(x_{j-1},x_j))e^{-\mu d(x_{j-1},x_j)} \leq \frac{\|F\|}{C_\mu} \min\{|\operatorname{supp}A|,|\operatorname{supp}B|\} \exp[-\mu d(A,B)] 
\end{equation*}

Define the projector,
$$P_{\Lambda}^{k}=\left\{\begin{array}{c}  P_\Lambda ,\mbox{ if $k=0$}\\
\mathbb{1}-P_\Lambda,\mbox{ if $k=1$} \end{array}\right.$$
where $P_\Lambda$ is the projector on $\Ker(I_\Lambda+I_\Lambda^*)$. 
From condition (\ref{structure}), we can rewrite the expression,
\begin{align*}
&\mathcal{J}_{\mathcal{P}_n} \leq 2\|A\|\ \|B\| \int_0^t \int_0^{t-s_1
}...\int_0^{t-\sum_{j=1}^{n-1}s_j}d\textbf{s}\    \Big|\Big|\prod_{\Lambda_j\in\mathcal{P}_n} \left(\exp\left[s_j L^c_{\mathcal{P}_n^{j-1}}\right] I_{\Lambda_j}/\|I_{\Lambda_j}\|_{cb}(P^0_{\Lambda_j} + P^1_{\Lambda_j}) \right)\Big|\Big|
\end{align*}
This equation consists of $2^n$ terms. For each string of bits of length $n-k+1$ $\{i_k,...,i_n\} \subset \{0,1\}^k$ for $k=1,...,n$, define the projector $$Q^{\{i_k,...,i_n\}}_{\mathcal{P}_n^{k-1}}=\prod_{\Lambda_j \in \mathcal{P}_n \backslash  \mathcal{P}_n^{k-1}}(\mathbb{1}-i_j P_{\Lambda_j})$$
From condition (\ref{structure}) we see that this is indeed a projector. Moreover using (\ref{structure}), we can rewrite our expression,
\begin{align*}
&\mathcal{J}_{\mathcal{P}_n}\leq 
 2\|A\|\ \|B\| \sum_{\{i_1,...,i_n\}  \subset \{0,1\}^n} \mathcal{R}^{\{i_1,...,i_n\}}_{\mathcal{P}_n} \mathcal{S}^{\{i_1,...,i_n\}}_{\mathcal{P}_n}
\end{align*}
with,
\begin{align*}
&\mathcal{R}^{\{i_1,...,i_n\}}_{\mathcal{P}_n}= \prod_{{\Lambda_k\in\mathcal{P}_n}}\left(\|I_{\Lambda_k}P^{i_k}_{\Lambda_k}\|_{cb}/\|I_{\Lambda_k}\|_{cb} \right) \\
&\mathcal{S}^{\{i_1,...,i_n\}}_{\mathcal{P}_n} =\int_0^t \int_0^{t-s_1
}...\int_0^{t-\sum_{j=1}^{n-1}s_j}d\textbf{s}\prod_{{\Lambda_k\in\mathcal{P}_n}} \Big|\Big|Q^{\{i_k,...,i_n\}}_{\mathcal{P}_n^{k-1}}\exp\left[s_k Q^{\{i_k,...,i_n\}}_{\mathcal{P}_n^{k-1}}L^c_{\mathcal{P}_n^{k-1}}Q^{\{i_k,...,i_n\}}_{\mathcal{P}_n^{k-1}}\right]Q^{\{i_k,...,i_n\}}_{\mathcal{P}_n^{k-1}}\Big|\Big|
\end{align*}
We use a variation of \citep[][Thm IX.3.1]{bhatia1997matrix}. 
Given $A,P\in \mathcal{M}_n(\mathbb{C})$ and  $P$ projector,
\begin{equation*}
\|P \exp[PAP]P\|\leq \Big|\Big|P \exp\left[P\frac{A+A^\dag}{2}P\right]P\Big|\Big|\end{equation*}
This equation and the definition of $Q^{\{i_k,...,i_n\}}_{\mathcal{P}_n}$, we see that we can bound each terms in $\mathcal{S}^{\{i_1,...,i_n\}}_{\mathcal{P}_n}$,
\begin{align*}
 \Big|\Big|Q^{\{i_k,...,i_n\}}_{\mathcal{P}_n^{k-1}}\exp\left[\frac{1}{2}s_k Q^{\{i_k,...,i_n\}}_{\mathcal{P}_n^{k-1}}(L^c_{\mathcal{P}_n^{k-1}}+L^{c*}_{\mathcal{P}_n^{k-1}})Q^{\{i_k,...,i_n\}}_{\mathcal{P}_n^{k-1}}\right]Q^{\{i_k,...,i_n\}}_{\mathcal{P}_n^{k-1}}\Big|\Big|
 \leq  \exp\left[-\nu s_k \sum_{j=k }^n i_j\right]
\end{align*}

Defining new variables,
\begin{align*}
u_1=s_1,~~ u_2=s_1+s_2,~~...~~,u_k=\sum_{j=1}^k s_j~~...
\end{align*}
We can rewrite $\mathcal{S}^{\{i_1,...,i_n\}}_{\mathcal{P}_n}$,
\begin{equation*}
\mathcal{S}^{\{i_1,...,i_n\}}_{\mathcal{P}_n} = \int_0^{t}\int_{u_1}^{t}...\int_{u_1+...+u_{n-1}}^{t}d\textbf{u}\   \exp{\left[-\sum_{j=1}^{n}u_j \sum_{k=j}^{n} i_k \right]}
\end{equation*}
and,
\begin{equation*}\| I_{\Lambda_k}P^{i_k}_{\Lambda_k}\|_{cb}/\|I_{\Lambda_
k}\|_{cb}\leq \left\{\begin{array}{c} \beta ,\mbox{ if $i_k=1$}\\
				\alpha ,\mbox{ if $i_k=0$} \end{array}\right.
\end{equation*}
Using these two bounds $\mathcal{J}_{\mathcal{P}_n}$ becomes,
\begin{align*}
\mathcal{J}_{\mathcal{P}_n}
&\leq 2\|A\|\ \|B\|\ \int_0^{t}\int_{u_1}^{t}...\int_{u_1+...+u_{n-1}}^{t} d\textbf{u}\  \left(\alpha+\beta\exp(-\nu u_1)\right)...\left(\alpha+\beta\exp(-\nu u_n)\right)\\
& = 2\|A\|\ \|B\|\frac{1}{(n-1)!}\ \int_0^{t}du_1\ \left(\alpha+\beta\exp(-\nu u_1)\right)\left(\alpha (t-u_1)+\beta \left(\frac{\exp(-\nu u_1)-\exp(-\nu t)}{\nu}\right)\right)^{n-1}\\
&= 2 \|A\|\ \|B\|\ \frac{1}{n!}\left[\alpha t+\beta \left(\frac{1-\exp(-\nu t)}{\nu}\right)\right]^n
\end{align*}
The first equality is the induction hypothesis, the case $n=1$ can be checked easily. Combining our results, we get our claim.
\end{proof}
\begin{Lemma}
For  $A$ and $I_Z$ with $Z \cap \Lambda_A =\emptyset$,
\begin{equation*}
\|\ I_Z \exp(tL)A\ \| \leq \sum_{n=1}^\infty \|L\|_\mu^n \left(\sup\limits_{\mathcal{P}_n}\mathcal{J}_{\mathcal{P}_n}\right) \sum\limits_{\mathcal{P}_n}\sum\limits_{x_0\in \Lambda_A}\prod_{\Lambda_j \in \mathcal{P}_n}\sum_{x_j \in \Lambda_j}\sum\limits_{x_{n+1} \in Z}F(d(x_{j-1},x_j)) e^{-\mu d(x_{j-1},x_j)}
\end{equation*}
\begin{proof}
The procedure is very similar to \cite{PhysRevLett.93.140402,2010arXiv1004.2086N}. Consider the generator $L_{\cap Z}$,
$$L_{\cap Z} =\sum_{\Lambda_1\cap Z\not=\emptyset}I_{\Lambda_1} $$
Using equation (\ref{exp2}) , we can then expand,

$$I_Z \exp(tL)A = \sum_{\Lambda_1\cap Z\not=\emptyset}\|I_{\Lambda_1}\|_{cb}\int_0^t I_Z\exp[s_0 L] I_{\Lambda_1}/\|I_{\Lambda_1}\|_{cb} \exp[(t-s_0)(L-L_{\cap Z})]A$$
We then continue this procedure for each term $I_{\Lambda_1} \exp[(t-s_0)(L-L_{\cap Z})]A$  for which $\Lambda_1 \cap \Lambda_A =\emptyset$. Let us say that after $k-1$ steps we have $\Lambda_k \cap \Lambda_A \not=\emptyset$, from our construction we see that we have built a path of $k$ steps from $Z$ to $A$ for which $\mathcal{J}_{\mathcal{P}_k(Z,A)}$ can be bounded by equation (\ref{Jpn}). For each $n$, we can then use the definition of $\|L\|_{\mu}$,
\begin{align*}
\left(\sup\limits_{\mathcal{P}_n}\mathcal{J}_{\mathcal{P}_n}\right)\sum_{\mathcal{P}_n}\prod_{\Lambda_j \in \mathcal{P}_n}\sum_{\Lambda_j}\|I_{\Lambda_j}\|_{cb}\leq \left(\sup\limits_{\mathcal{P}_n}\mathcal{J}_{\mathcal{P}_n}\right) \sum\limits_{\mathcal{P}_n} \|L\|_{\mu}^n \sup\limits_{x_0\in \Lambda_A}\prod_{\Lambda_j \in \mathcal{P}_n}\sum_{x_j \in \Lambda_j}\sum\limits_{x_{n+1} \in Z}F(d(x_{j-1},x_j))  e^{-\mu d(x_{j-1},x_j)}
\end{align*}
\end{proof}
\end{Lemma}
We can then use this procedure for $[B,.]$.

\section{Examples}
\subsection{Specific Subclass}
We present a method to construct examples that satisfy the equations (\ref{structure}) for the case $P_{X}=\frac{\mathbb{1}_{X}}{\operatorname{Tr}({\mathbb{1}_{X}})}\operatorname{Tr}_X[.]$ and $X$ consists of at most two sites. For this construction $\alpha =0$ and the Lieb-Robinson velocity goes asymptotically to zero. 

Consider $2n$ trace-preserving complete positive operators $\Gamma^j,\Phi^j:B(\mathcal{C}^d)\to B(\mathcal{C}^d)$ and $c_1,...,c_n \in (0,1]$ with $\sum_k c_k =1$, so that $\Gamma_j^*$ and $\Phi_j^*$ are not trace-preserving and $\sum_{j} c_j \Gamma^j \otimes \Phi^j$ only has the uniform state as fixed point. We can then easily see that interaction terms $I_{i,j}$ of the form,
$$I_{i,j}[.]= \sum_k c_k \Gamma^k_{i}\otimes \Phi^k_{j}[.] -[.] $$
satisfy the structural equation (\ref{structure}).

Let us give some concrete examples. 
For $A \in  B(\mathcal{C}^2)$, $\lambda, t \in [-1,1]$, consider the following trace-preserving map 
$\Phi(\lambda,t):B(\mathcal{C}^2)\to B(\mathcal{C}^2)$.
$$\Phi(\lambda,t):\left(\begin{array}{cc}A_{11}& A_{12}\\ A_{21}& A_{22}\end{array}\right)\to \left(\begin{array}{cc}  \frac{1}{2}(A_{11}+A_{22})&\frac{t}{2} (A_{11}+A_{22})+\lambda A_{12} \\\frac{t}{2} (A_{11}+A_{22})+\lambda A_{21} & \frac{1}{2}(A_{11}+A_{22})\end{array}\right) $$
From \cite{2001quant.ph..1003R}, the condition
 $|\lambda| + |t|<1$ implies complete positivity.
We can then construct the following interaction terms,
$$ I_{i,j}[.]= \frac{1}{3} \Phi_i(\lambda_1,t_1) \otimes \Phi_{j}(\lambda_2,t_2)[.] +\frac{1}{3} \Phi_i(\lambda_3,s)\otimes \Phi_{j}(\lambda_4,s)[.]+\frac{1}{3} \Phi_j(\lambda_5,r)\otimes \Phi_{j}(\lambda_6,-r)[.]-[.]$$
with $t_1+r+s=t_2-r+s=r^2-s^2=0$. 

As a second example consider the map,
$\Psi(\lambda,t):B(\mathcal{C}^2)\to B(\mathcal{C}^2)$.
$$\Psi(\lambda,t):\left(\begin{array}{cc}A_{11}& A_{12}\\ A_{21}& A_{22}\end{array}\right)\to \left(\begin{array}{cc}  \frac{1+
\lambda +t}{2}A_{11}+ \frac{1-\lambda +t}{2}A_{22}& 0 \\0 & \frac{1-\lambda -t}{2}A_{11}+ \frac{1+\lambda-t}{2}A_{22}\end{array}\right) $$
with $-1<t,\lambda< 1,\ 1-|\lambda|-|t|>0$.
Given $-1<r_1,r_2,s_1,s_2, t_1,t_2, u_1,u_2<1$, so that,
$$\operatorname{det}\left(\begin{array}{cccc}(1+r_1)(1+r_2)& (1+s_1)(1+s_2) &  (1+t_1)(1+t_2) &(1+u_1)(1+u_2)\\ (1+r_1)(1-r_2)& (1+s_1)(1-s_2) &  (1+t_1)(1-t_2) &(1+u_1)(1-u_2)\\
(1-r_1)(1+r_2)& (1-s_1)(1+s_2) &  (1-t_1)(1+t_2) &(1-u_1)(1+u_2)\\
(1-r_1)(1-r_2)& (1-s_1)(1-s_2) &  (1-t_1)(1-t_2) &(1-u_1)(1-u_2)\end{array}\right)\not=0$$
And the interactions are given by,
 \begin{align*}
 I_{i,j}[.]= & c_1 \Psi_i(\lambda_1,r_1) \otimes \Psi_{j}(\lambda_2,r_2)[.] +c_2 \Psi_i(\kappa_1,s_1)\otimes \Psi_{j}(\kappa_2,s_2)[.]+\\
 & c_3 \Psi_i(\mu_1,t_1)\otimes \Psi_{j}(\mu_2,-t_2)[.]+c_4 \Psi_i(\nu_1,u_1)\otimes \Psi_{j}(\nu_2,u_2)[.]-[.]
 \end{align*}

Let us finally remind the reader that a maximum range of the interaction is necessary, $\forall I_{i,j},\ d(i,j)\leq R <\infty$. This condition implies an exponential decay in distance in equation (\ref{eventh}). Long range interactions are possible in some cases, but mostly imply slower decays. Further discussion can be found in \cite{KomatHastings}.
\subsection{Local Covariance}

In the previous section, we gave a method for the construction of non-trivial examples. These examples give rise to an approximate event
 horizon or in other words are quasi-ultra-local.
An interesting question is how do ultra-local dynamics look like. In the case of automorphisms, we would need interaction terms to commute with each other. In this section we give an interesting class of ultra-local dynamics with non-commutating interaction terms.

We use the concept of local covariance to introduce this particular class. It can then be combined with the previous examples to construct additional quasi-ultra-local dynamics.

Given some group $G$ with a unitary representation $V_g:\mathcal{H}_x\to \mathcal{H}_x$, $\forall x$. We say that $I_X$ is local covariant if, $\forall i\in X$,
$$I_X[(\mathbb{1}_{X\backslash \{i\}}\otimes V_g)(.)(\mathbb{1}_{X\backslash \{i\}}\otimes V_g^\dag)]=(\mathbb{1}_{X\backslash \{i\}}\otimes V_g)I_X[.](\mathbb{1}_{X\backslash \{i\}}\otimes V_g^\dag)$$
Additionally, we say $\operatorname{Ker} (I_X + I_X^\dag)$ is invariant with respect to the representation $\{V_g\}$ of a group $G$ if $\forall A\in \operatorname{Ker}  (I_X + I_X^\dag),\forall i\in X$,
$$(\mathbb{1}_{X\backslash \{i\}}\otimes V_g) A (\mathbb{1}_{X\backslash \{i\}}\otimes V_g^\dag) = A$$
in other words $\{V_g\}$ is a local gauge symmetry of $\operatorname{Ker}  (I_X + I_X^\dag)$.

For every amenable group $G$ with invariant mean $\mu$ and representation $\{V_g\}$, we can define a projector $P_X$ on $\mathcal{A}_X$ onto the vector space with this representation as local gauge symmetry.
\begin{equation} 
P_X: \mathcal{A}_{X}\to \mathcal{A}_{X}:A\to \prod_{i\in X} P_i [A]
\label{proj}
\end{equation}
with,
$$P_i: A\to \sum_{g\in G}\mu(g) (V_g \otimes \mathbb{1}_{X \backslash \{i\}}) A (V_g^\dagger \otimes \mathbb{1}_{X \backslash \{i\}})$$
It can be easily checked, that equations (\ref{structure}) are satisfied when $I_X$ are locally covariant. Additionally we can see that every interaction terms $I_X$ commute with the projectors $P_Y$. This implies ultra-locality, $\forall A\in\mathcal{A}_{\Lambda_A}$,
$$\exp \left[t\sum_X I_X \right]A=\exp\left[t\sum_{X\cap \Lambda_A\not=\emptyset}I_X \right]A$$
Covariance properties of dynamics has been largely studied \cite{Holevo1993211,402000768,1064-5632-59-2-A10}.

\section{Further Discussions}

\subsection{Localization}

In \cite{BOE}, systems with time dependent disorder were studied. The time dependence was chosen so that the dynamics would be equivalent to one generated by a local Lindblad generator. Specifically the generator was of the form, 
\begin{equation}
L[.] =\sum_{j}i[h_{j,j+1},.]-\gamma\frac{\mathbb{1}_{j}}{\operatorname{Tr}({\mathbb{1}_{j}})}\operatorname{Tr}_j[.] 
\label{example}
\end{equation}
The second terms are local depolarizations.
They showed that when $\gamma$ scaled with the size of the system, the Lieb Robinson velocity is suppressed and we get a result similar to Anderson localization \cite{Anderson}.

Originally the Lieb Robinson bound was used to prove the existence of the thermodynamic limit. For the proof, it is important for the bound to be independent of the size of the system. Therefore the result in \cite{BOE} can only be seen as a hint for effective velocity. Indeed, the authors conjectured and showed numerically that for small $\gamma$, localization should still appeared. When the local depolarization is too weak, the information propagates again. 
Using the method developed in this article, we prove their conjecture.

Consider the following local dissipative dynamics,
$$L =L_{int}+L_{diss}=\sum_{\Lambda_j}I_{\Lambda_j}+\sum_j D_j$$
with $|\Lambda_j|>1$.
This notation means that $D_j$ are one-point local and play the role of dissipation, while $I_{\Lambda_j}$ are the interaction terms that are responsible for the hopping.

Denote $P_{\lambda_j}$ the projector onto the kernel of $I_{\Lambda_j}+I_{\Lambda_j}^\dag +\sum_{k\in \Lambda_j} (D_k+D^\dagger_k) $.
Define then the projector onto the kernel of the hermitian part all interactions except the first $k$ interactions along a path of length $n$,
$P^c_{\mathcal{P}_n^{k}}=\prod_{\Lambda_j\not \in \mathcal{P}_n^{k}}P_{\Lambda_j}$

Assume now the following,
\begin{align}
\forall \Lambda_j,\Lambda_k,~~ P_{\Lambda_j}P_{\Lambda_k}=P_{\Lambda_k}P_{\Lambda_j} \\
\forall \mathcal{P}_n,j,~~P^c_{\mathcal{P}_n^{j}} L^c_{\mathcal{P}_n^{j}}P^c_{\mathcal{P}_n^{j}}= P^c_{\mathcal{P}_n^{j}} L^c_{\mathcal{P}_n^{j}}
\label{assumption}
\end{align}

The method of Lemma 1 can then be repeated except only the interactions $I_X$ are taken out of the exponentials and not the one-point dissipative terms $D_j$. This leads us again to the equations,
\begin{equation*}
\|[B, \exp(tL)A]\| \leq  \sum_{n=1}^\infty \|L_{int}\|_\mu^n \left(\sup\limits_{\mathcal{P}_n}\mathcal{J}_{\mathcal{P}_n}\right) \sum\limits_{\mathcal{P}_n}\sum\limits_{x_0\in \Lambda_B}\prod_{\Lambda_j \in \mathcal{P}_n}\sum_{x_j \in \Lambda_j}\sum\limits_{x_{n+1} \in \Lambda_A}F(d(x_{j-1},x_j)) e^{-\mu d(x_{j-1},x_j)}
\end{equation*}
with,
\begin{align}
&\mathcal{J}_{\mathcal{P}_n} = \Big|\Big|\int_0^t \int_0^{t-s_1
}...\int_0^{t-\sum_{j=1}^{n-1}s_j} d\textbf{s}\ [B,.]  \  \prod_{j=1}^{n}\left(\exp\left[s_j L^c_{\mathcal{P}_n^{j-1}}\right]  I_{\Lambda_j}/\|I_{\Lambda_j}\|_{cb}\right) \exp[(t-\sum_{j=1}^{n}s_j) L^c_{\mathcal{P}_n^n}]A \Big|\Big|
\label{Jpn}
\end{align}

We can then insert again the projectors,

\begin{align*}
\mathcal{J}_{\mathcal{P}_n} = &\Big|\Big|\int_0^t \int_0^{t-s_1
}...\int_0^{t-\sum_{j=1}^{n-1}s_j} d\textbf{s}\ [B,.]  \  \prod_{j=1}^{n}\left(\exp\left[s_j L^c_{\mathcal{P}_n^{j-1}}\right](\mathbb{1}-P^c_{\mathcal{P}_n^{j-1}})  I_{\Lambda_j}/\|I_{\Lambda_j}\|_{cb}\right)\\ &\exp[(t-\sum_{j=1}^{n}s_j) L^c_{\mathcal{P}_n^n}](\mathbb{1}-P^c_{\mathcal{P}_n^{n}})A \Big|\Big|
\end{align*}

Since, $[L^c_{\mathcal{P}_n^k},P^c_{\mathcal{P}_n^{k}}]=0$, $\forall \mathcal{P}_n,k$ and using the assumption (\ref{assumption}), 
$$\|(\mathbb{1}-P^c_{\mathcal{P}_n^{j}})\exp\left[s_j(\mathbb{1}-P^c_{\mathcal{P}_n^{j}}) (L^c_{\mathcal{P}_n^{j}}+L^{c\dag}_{\mathcal{P}_n^{j}})(\mathbb{1}-P^c_{\mathcal{P}_n^{j}})/2\right](\mathbb{1}-P^c_{\mathcal{P}_n^{j}}) \|\leq \exp(-\lambda s_j) $$

Therefore,
$$\mathcal{J}_{\mathcal{P}_n}\leq \exp(-\lambda t) \frac{\|L_{int}\|_{\mu}^n t^n}{n!}$$

From which follows,
$$\|[A, \Gamma_t(B)] \|\leq C' \|A\|\ \|B\| \exp\left( \frac{-\lambda t + v t - d(A,B)}{\xi} \right)$$
Generally $\lambda < v$, however,
 in the case of one-point dissipative terms, it is possible to make $\lambda > v$.
In the case of our example (\ref{example}), we now show that the assumptions (\ref{assumption}) are satisfied and  $\lambda= \gamma$ and $v\leq 2 \max_j\|h_{j,j+1}\|$.

\begin{proof}
Notice first $P_{j,j+1}=\frac{\mathbb{1}_{j,j+1}}{\operatorname{Tr}({\mathbb{1}_{j,j+1}})}\operatorname{Tr}_{j,j+1}\otimes \operatorname{Id}_{\Lambda\backslash \{j,j+1\}}$.
We can take some irreducible representation $\{U_{g}\}$ of a group $G$ with invariant mean $\mu$ that lives on site $j$.
Since,
$$\frac{\mathbb{1}_{j,j+1}}{\operatorname{Tr}({\mathbb{1}_{j,j+1}})}\operatorname{Tr}_{j,j+1}[.] =\sum_{g,h} \mu(g) \mu(h) U_{g}\otimes  U_{h} [.]U_{g}^\dagger \otimes  U_{h}^\dagger $$
from the first assumption can be checked.

Notice then that,
$$(L^c_{\mathcal{P}_n^{j}}+L^{c\dag}_{\mathcal{P}_n^{j}})=-2\gamma  \sum_{k=x_A}^{x_B} P_j$$

from this then follows $\lambda = \gamma $
\end{proof}

Therefore we see that there is no hopping when $\gamma> 2\|h_{j,j+1}\|
, \forall j$.

\subsection{Exponential Clustering Theorem}

 In \cite{PhysRevLett.104.190401}, A general method was given for systems with unique invariant state and a dissipative gap.  For gapped frustration free-Hamiltonian, proving correlation properties is quite easy. Similarly, we give a simple proof of the correlation properties of frustration free local Markovian dynamics with a dissipative gap.

Consider a frustration free local Markovian dynamics $L=\sum_X I_X$ with unique invariant state $\rho_\beta$. By which we mean, $L^* \rho_\beta = I_X^* \rho_\beta =0$, $\forall X$. We say that $L$ has a dissipative gap, if there is some $\lambda>0$, so that $\forall A\in\mathcal{A}_\lambda$, with $\operatorname{Tr}(A)=0$,  
$$\|\exp(tL)A \|\leq \|A\| \exp(-\lambda t)$$
For every local Markovian dynamics, the usual Lieb-Robinson bound can be shown,
\begin{equation}
\|I_X\exp[tL]A]\|\leq C \|A\|\ \|I_X\|\ \exp\left[\frac{vt-d(A,X)}{\xi}\right]
\label{LRB2}
\end{equation}
Notice that it should be clear that for the case that dynamics are frustration-free, the following bound can be shown,
\begin{equation}
\|I_X^*\exp[tL^*][A \rho_\beta]\|\leq C \|A\|\ \|I_X\|\ \exp\left[\frac{vt-d(A,X)}{\xi}\right]
\label{LRB2}
\end{equation}
We now show that for every local observable $A\in\mathcal{A}_{\Lambda_A}, B\in\mathcal{A}_{\Lambda_B}$, 
$$|\operatorname{Tr}( \rho_\beta A B)-\operatorname{Tr}( \rho_\beta A)\operatorname{Tr}( \rho_\beta B)|\leq \|A\|\ \|B\|\ \exp\left(-d(A,B)/(\xi +v/\lambda)\right)$$

\begin{proof}
Without loss of generality, we can take $\operatorname{Tr}( A)=\operatorname{Tr}( \rho_\beta B)=0$.
Notice first that from the Lieb-Robinson bound (\ref{LRB2}) follows,
$$\|\exp(tL^*)[B\rho_\beta] -\exp(t L^*-t\sum_{X\cap \Lambda_A \not=\emptyset} I_X^* )[B\rho_\beta]\|\leq C' \max_X\|I_X\|\ \|B\| \exp( (vt-d(A,B))/\xi)$$

Indeed,
\begin{align*}
\|& \exp(tL^*)[B\rho_\beta] -\exp(t L^*-t\sum_{X\cap \Lambda_A \not=\emptyset} I_X^* )[B\rho_\beta]\|\\
&= \|
\int_0^t ds \exp( (t-s)L^*- (t-s)\sum_{X\cap \Lambda_A \not=\emptyset}I_X^*) \sum_{X\cap \Lambda_A \not=\emptyset}I_X^* \exp(sL^*) [B\rho_\beta]\|\\
&\leq C \int_0^t ds \max_X\|I_X\|\ \|B\| \exp( (vt-d(A,B))/\xi) 
\leq C' \max_X\|I_X\|\ \|B\| \exp( (vt-d(A,B))/\xi)
\end{align*}

Our claim then follows from,
\begin{align*}
|\operatorname{Tr}( \rho_\beta B A)|=|\operatorname{Tr}(\rho_\beta B \exp(t L-t\sum_{X\cap \Lambda_A \not=\emptyset} I_X )[A])=|\operatorname{Tr}( A \exp(t L^*-t\sum_{X\cap \Lambda_A \not=\emptyset} I_X^* )[\rho_\beta B] )|\\
\leq \|B\|\ \|\exp(tL)A\| +\|A\|\ \| \exp(tL^*)[\rho_\beta B]-\exp(t L^*-t\sum_{X\cap \Lambda_A \not=\emptyset} I_X^* )[\rho_\beta B]\|)
\end{align*}
\end{proof} 

It should be clear that the correlation length of systems with dynamics satisfying equation \ref{structure} is zero.

\subsection{Conclusion}

In this article, we improved the Lieb-Robinson bound for a class of local dissipative dynamics. For this class, the bound reflects the dissipative character of the dynamics. The bound proven in \cite{2011arXiv1103.1122N} only showed the spreading of local perturbation due to the local property of the dynamics. We showed, however, that when the interactions satisfy additional structural equations (\ref{structure}), this spreading can is slowed down due to the dissipative property.

Further extensions of this class could be of interest, especially the change of the bound under a local perturbation of the class.

In the last section, we also showed that our method is not only restricted to the proof of our particular class. We used to prove the conjecture in (REF). In order to find additional classes one has to find the right subspace. Subspace with constant spacing in the rate of dissipation lead to the exponential decrease of the velocity. If the spacing were to decrease, we get a different behavior fo the velocity. The difficulty now lies in constructing concrete examples.

\section*{ACKNOWLEDGMENTS}
We acknowledge financial support by the FWF project CoQuS No. W1210‐N1. The author thanks James D. Whitfield for helpful discussions.

\end{document}